\documentclass[12pt]{article}

\usepackage{hyperref}

\usepackage{amsmath} 
\usepackage{amssymb}
\usepackage{amsfonts}
\usepackage{amsthm}
\usepackage{graphicx}

\newtheorem{theorem}{Theorem}
\newtheorem{lemma}{Lemma}
\newtheorem{corollary}{Corollary}

\theoremstyle{definition}

\begin{document}
	\begin{center}
		\Large
	\textbf{Long-time Markovianity of multi-level systems in the rotating wave approximation}\footnote{This work is supported by the Russian Science Foundation under grant 17-71-20154.}	
	
		\large 
		\textbf{A.E. Teretenkov}\footnote{Department of Mathematical Methods for Quantum Technologies, Steklov Mathematical Institute of Russian Academy of Sciences,
			ul. Gubkina 8, Moscow 119991, Russia\\ E-mail:\href{mailto:taemsu@mail.ru}{taemsu@mail.ru}}
		\end{center}
		
			\footnotesize
			For the model of a multi-level system in the rotating wave approximation we obtain the corrections for a usual weak coupling limit dynamics by means of perturbation theory with Bogolubov-van Hove scaling. It generalizes our previous results on a spin-boson model in the rotating wave approximation. Additionally, in this work we take into account some dependence of the system Hamiltonian on the small parameter. We show that the dynamics is long-time Markovian, i.e. after the bath correlation time all the non-Markovianity could be captured by the renormalization of initial condition and correlation functions. 
			\normalsize

	\section{Introduction}
	
	There is intense discussion in literature about different approaches to definition and characterization of quantum Markovianity (see \cite{Li2018} for review). This is important due to modern both theoretical and applied interest in the non-Markovian phenomena in the open quantum systems (see e.g. \cite{Li2020a, Li2020b, Chruscinski2019} for recent reviews). Most of the known measures of non-Markovianity think about Markovianity \cite{Breuer09, Gullo14, Rivas14, Bae2016, Haikka11} as of some property which is global in time. But a few  works \cite{Petrosky2002, Teretenkov2020nonpert} suggest that it is more natural to speak about some initial time (Zeno time) of order of bath correlation time before which the dynamics is surely highly non-Markovian and only after that it becomes Markovian. In \cite{Teretenkov2020nonpert} we have called such a behaviour long-time Markovian and have shown that it could be naturally captured by perturbation theory with Bogolubov-van Hove scaling. Bogolubov-van Hove scaling does not only insert the small parameter $ \lambda $ before the coupling constant but also rescales the time as $ t \rightarrow \lambda^{-2} t $. It allows one to separate the time-scale on which the Markovian behavior occurs from the time scale of order of the bath correlation time which becomes of order of $ \lambda^{2} $ after the scaling. In \cite{Teretenkov2020nonpert} we have considered the simplest model, namely, the spin-boson in the rotating wave approximation (RWA). Here we generalize the main results of \cite{Teretenkov2020nonpert} to the multi-level model considered in \cite{Teretenkov19} and \cite{Teretenkov20}. 
	
	In Section~\ref{sec:intDiffEq} we recall the results from \cite{Teretenkov19, Teretenkov20} in such a manner which is useful for the further parts of the article. In Section~\ref{sec:expansion} we obtain the first asymptotic correction to the dynamics obtained in the Bogolubov-van Hove limit. Inspired by the unified  Gorini--Kossakowski--Sudarshan--Lindblad (GKSL) quantum master equation approach \cite{Trushechkin2021} we do not only directly generalize the results of \cite{Teretenkov2020nonpert} here, but also take into account terms of order $ \lambda^2 $ in the system Hamiltonian. We show that the corrected dynamics of the reduced density matrix after the correlation time could be described by a semigroup, but the initial condition should be renormalized. This leads to the corrected master equation with a time-independent generator, which is of the GKSL form for sufficiently small $ \lambda $.  In Section~\ref{sec:correlation} we show that if one defines Markovianity in terms of the system correlation functions, then it leads to the semigroup property for the dynamical map describing the reduced density matrix. Thus, strictly speaking, our dynamics is not Markovian in the sense of system correlation functions even asymptotically, but all the non-Markovianty could be absorbed in the  renormalization of the correlation functions. So we establish long-time Markovian properties of the reduced dynamics for this model in the same sense as in \cite{Teretenkov2020nonpert}.
	
	\section{Integro-differential Schroedinger equation}
	\label{sec:intDiffEq}
	
	We consider the model of a multi-level system interacting with several reservoirs from \cite{Teretenkov19}. So let us recall its definition and the main results which we use in this paper. We consider the evolution in the Hilbert space
	\begin{equation*}
		\mathcal{H} \equiv (\mathbb{C}\oplus\mathbb{C}^{N}) \otimes \bigotimes\limits_{i=1}^N \mathfrak{F}_b(\mathcal{L}^2(\mathbb{R})),
	\end{equation*}
	Here $  \mathbb{C}\oplus\mathbb{C}^{N}$ is an $ (N+1) $-dimensional Hilbert space with a pointed one-dimensional subspace which corresponds to the degrees of freedom of the $ (N+1) $-level system. Let $ | i \rangle , i = 0,  1, \ldots,  N $ be an orthonormal basis in such a space and $ | 0 \rangle $ correspond to the pointed subspace.  $ \mathfrak{F}_b(\mathcal{L}^2(\mathbb{R})) $ are bosonic Fock spaces which describe the reservoirs. Let $ | \Omega \rangle $ be a vacuum vector for the reservoirs. Let us also introduce the creation and annihilation operators which satisfy the canonical commutation relations: $ [b_{k,i}, b_{k',j}^{\dagger}] = \delta_{ij} \delta (k - k')$, $[b_{k,i}, b_{k',j}] = 0 $, $ b_{k,i} | \Omega \rangle = 0$.
	
	We consider the system Hamiltonian of the general form with the only requirement that it vanishes on the ground state. Namely, $ \hat{H}_S  =  0 \oplus H_S $, where $ H_S $  is an $ N \times N $ (Hermitian) matrix. The reservoir Hamiltonian is a sum of similar Hamiltonians of the free bosonic fields (with the same dispersion relation $ \omega(k) $)
	\begin{equation*}
		\hat{H}_B =\sum_{i=1}^N \int \omega(k) b_{i}^{\dagger}(k)  b_{i}(k)  d k .
	\end{equation*}
	The interaction is described by the following Hamiltonian
	\begin{equation*}
		\hat{H}_I = \sum_i \int \left(  g_{ k}^*  | 0 \rangle \langle i| \otimes b_{k,i}^{\dagger}+   g_{k}  | i \rangle \langle 0 | \otimes b_{k,i} \right) d k.
	\end{equation*}
	
	Let us denote the unitary evolution of the density matrix
	\begin{equation*}
		\rho(t) = e^{- i \hat{H} t} \rho(0) e^{ i \hat{H} t}
	\end{equation*}
	with the Hamiltonian $ \hat{H} = \hat{H}_S \otimes I + I \otimes \hat{H}_B + \hat{H}_I $ and factorized initial condition
	\begin{equation}\label{eq:initCond}
		\rho(0) =  \rho_S(0) \otimes  | \Omega \rangle \langle \Omega |.
	\end{equation}
	We prefer to  consider the evolution of the density matrix in the interaction representation, so let us define
	\begin{equation*}
		\rho_I(t) \equiv e^{ i (\hat{H}_S \otimes I + I \otimes \hat{H}_B) t} \rho(t)  e^{- i (\hat{H}_S \otimes I + I \otimes \hat{H}_B) t}.
	\end{equation*}
	Moreover, we are interested in the reduced density dynamics, so let us introduce
	\begin{equation*}
		\rho_{SI}(t) \equiv \operatorname{Tr}_{B} \rho_I(t).
	\end{equation*}
	This model and its particular cases are widely used as a test model for many approaches to the open quantum systems \cite{Friedrichs48, Garraway96, Garraway97, Garraway97a, Dalton01, Garraway06, Luchnikov19, Teretenkov19m}. But let us remark  that it omits non-RWA effects which could be important for real physical systems \cite{Tang13, Fleming10, Trubilko2020}.
	
	Let us summarize the results of \cite[Corollary 1]{Teretenkov19} and \cite[Theorem 1]{Teretenkov20} in the interaction representation by the following theorem.
	
	\begin{theorem}\label{th:previousResults}
		Let the integral
		\begin{equation*}
			G(t) = \int |g(k)|^2 e^{-i \omega(k) t} dk
		\end{equation*}
		converge for all $ t \in \mathbb{R}_+ $ and define the continuous function, then for pure initial system state $  \rho_S(0)  $ correspondent to the state vector of the form $ \psi_0(0) \oplus  | \psi (0) \rangle  $ one has
		\begin{equation}\label{eq:solIntPic}
			\rho_{SI}(t) =
			\begin{pmatrix}
				1 - || \psi (t) ||^2 & \psi_0(0) \langle \psi (t)|\\
				\psi_0^*(0) | \psi (t) \rangle  & | \psi (t) \rangle   \langle \psi (t)|
			\end{pmatrix},
		\end{equation}
		where $ | \psi (t) \rangle $ is the solution of the integro-differential equation
		\begin{equation}\label{eq:integroDiffInt}
			\frac{d}{dt} | \psi_I (t) \rangle = - \int_{0}^{t} ds \; G(t-s) e^{i H_S (t-s)} | \psi_I (s) \rangle
		\end{equation}
		with the initial condition $ | \psi (t) \rangle|_{t=0} = | \psi (0) \rangle$.
	\end{theorem}
	
	If $ \psi_0(0) = 0 $, i.e. the initial system state is excited, then Eq. \eqref{eq:integroDiffInt} describes the dynamics of excited state  dynamics instead of the Schroedinger equation, so we call it the integro-differential Schroedinger equation. And similar to the Schroedinger equation  the solution $ | \psi_I (t) \rangle $ could be represented \cite[Sec. 2.3]{Burton05} as $ | \psi_I (t) \rangle = V(t) | \psi_I (0) \rangle  $, where $ V(t) $ is an $ N \times N $ matrix which is a (unique) solution of 
	\begin{equation}\label{eq:integroDiffProp}
		\frac{d}{dt} V(t) = - \int_{0}^{t} ds \; G(t-s) e^{i H_S (t-s)} V(s)
	\end{equation}
	with the initial condition $ V(0) = I $ (but, generally, $ V(t) $ is not unitary). Similar to \cite{Teretenkov19OnePart} (but with a different notation) let us represent $ \rho_{S}(0) $ in the block form:
	\begin{equation}\label{eq:initRhoS}
		\rho_{S}(0) = 
		\begin{pmatrix}
			(\rho_{S}(0))_{00} & (\rho_{S}(0))_{0e} \\
			(\rho_{S}(0))_{e0} & (\rho_{S}(0))_{ee}
		\end{pmatrix},
	\end{equation}
	where $ (\rho_{S}(0))_{00} $ is just a $ 00 $-element of the density matrix $ \rho_{S}(0) $ ($ 1 \times 1 $ block), $ (\rho_{S}(0))_{e0} $ is a vector consisting of  $j0 $-elements $ j = 1, \ldots, N $ of the  density matrix $ \rho_{S}(0) $  ($ N \times 1 $ block), $ (\rho_{S}(0))_{0e} = ((\rho_{S}(0))_{e0} )^+ $  ($ 1 \times N $ block) and $ (\rho_{S}(0))_{ee} $ is the $ N \times N $ matrix formed by $jk $-elements $ j, k = 1, \ldots, N $  of the  density matrix $ \rho_{S}(0) $  ($ N \times N $ block). (The subscript ''$ e $'' stands for ''excited'' as our physical interpretation regards $ |j\rangle, j = 1, \ldots, N$ as excited  states.) If  blocks of the matrix $ \rho_{SI}(t) $ in the representation \eqref{eq:solIntPic} have the same structure, so we use the same notation for them. Namely, taking into account $ | \psi_I (t) \rangle = V(t) | \psi_I (0) \rangle  $ we have  $  (\rho_{S}(0))_{ee} = | \psi (t) \rangle   \langle \psi (t)| = V(t) | \psi (0) \rangle   \langle \psi (0)| V^+(t) =  V(t)(\rho_{S}(0))_{ee} V^+(t) $, $ (\rho_{SI}(t))_{e0} = V(t) (\rho_{S}(0))_{e0}$ and $ (\rho_{S}(0))_{00} = 1 - || \psi (t) ||^2 =  |\psi_0(0)|^2 + || \psi (0) ||^2 - || \psi (t) ||^2 = (\rho_{SI}(0))_{gg} + \mathrm{Tr} \; ((\rho_{SI}(0))_{ee}-V(t)(\rho_{SI}(0))_{ee}  V^{+}(t))$. Hence, one could represent \eqref{eq:solIntPic} in the form
	\begin{equation}\label{eq:solByProp}
		\rho_{SI}(t) = 
		\begin{pmatrix}
			(\rho_{SI}(0))_{gg} + \mathrm{Tr} \; ((\rho_{SI}(0))_{ee}-V(t)(\rho_{SI}(0))_{ee}  V^{+}(t)) &  (\rho_{SI}(0))_{ge} V^{+}(t) \\
			V(t) (\rho_{SI}(0))_{eg} & V(t)(\rho_{SI}(0))_{ee}  V^{+}(t)
		\end{pmatrix}
	\end{equation}
	for pure initial state $ \rho_{SI}(0) = \rho_{S}(0) $. This representation for $ \rho_{SI}(t) $ is linear in $ \rho_{SI}(0) $, so as an arbitrary $ \rho_{S}(0) $ could be represented as a convex combination in pure states, so formula \eqref{eq:solIntPic} is held for an arbitrary initial state $ \rho_{S}(0) $ as well. So we have obtained the following corollary of Th.~\ref{th:previousResults}.
	
	\begin{corollary}
		Let  $ G(t) $ have the same definition and properties as in Th.~\ref{th:previousResults}.  Let the initial reduced density matrix $ \rho_{S}(0) $ from \eqref{eq:initCond} have form \eqref{eq:initRhoS}. Then $ \rho_{SI}(t) $ could be defined by \eqref{eq:solByProp}, where $ V(t) $ is defined by integro-differential equation \eqref{eq:integroDiffProp} with the initial condition $ V(0) = I $.
	\end{corollary}
	
	Hence, the dynamics of the reduced density matrix is fully defined by $ V(t) $ and we can concentrate on the analysis of $ V(t) $ in the next section to understand dynamical properties of the reduced density matrix.
	
	\section{Expansion with Bogolubov-van Hove scaling}
	\label{sec:expansion}
	
	We are going to capture the asymptotic behavior of \eqref{eq:integroDiffProp}, so we introduce a small parameter $ \lambda $ before the interaction $ \hat{H}_I \rightarrow \lambda \hat{H}_I$, which leads to the squared small parameter in the  bath correlation function $ G(t) \rightarrow \lambda^2 G(t) $. We additionally assume that
	\begin{equation*}
		H_S = H_S^{(0)} + \lambda^2 H_S^{(2)}.
	\end{equation*}
	
	Hence, Eq. \eqref{eq:integroDiffProp} takes the form
	\begin{equation*}
		\frac{d}{dt} V_{\lambda}(t) = - \lambda^2 \int_{0}^{t} ds \; G(t-s) e^{i (H_S^{(0)} + \lambda^2 H_S^{(2)}) (t-s)} V_{\lambda}(s)
	\end{equation*}
	
	We also use the Bogolubov-van Hove scaling $ t \rightarrow \lambda^{-2} t$ as it is used for most mathematically strict derivations of weak coupling master equations \cite{Davies1974, Accardi2002, davies1976quantum, Pechen2002, Pechen2004, Accardi1990}.   Namely, we will expand the function 
	\begin{equation}\label{eq:defOfW}
		W_{\lambda}(t) = V_{\lambda}(\lambda^{-2} t) 
	\end{equation}
	
	The next result is actually very simple, but our experience shows that the small parameter $ \lambda $ appears in it in a bit contraintuitive way, so we give explicit proof of it.
	\begin{lemma}
		The function $ W_{\lambda}(t) $ satisfies
		\begin{equation}\label{eq:Weq}
			\frac{d}{dt} W_{\lambda}(t) = - \int_{0}^{t} ds \; \frac{1}{\lambda^2} G\left(\frac{t-s}{\lambda^2}\right) e^{i (\lambda^{-2} H_S^{(0)} +  H_S^{(2)}) (t-s)} W_{\lambda}(s)
		\end{equation}
		with initial condition $ W_{\lambda}(0) = I $.
	\end{lemma}
	
	\begin{proof} Let us directly differentiate $  W_{\lambda}(t) $ using Eq. \eqref{eq:integroDiffProp} and change the variable $ s \rightarrow \lambda^{-2} s$
		\begin{align*}
			\frac{d}{dt} W_{\lambda}(t) &= \frac{d}{dt} V_{\lambda}(\lambda^{-2} t) = \lambda^{-2}  V_{\lambda}'(\lambda^{-2} t) =  -  \int_{0}^{\lambda^{-2}t} ds \; G(\lambda^{-2}t-s) e^{i (H_S^{(0)} + \lambda^2 H_S^{(2)}) (\lambda^{-2} t-s)} V_{\lambda}(s) =\\
			&= - \int_{0}^{t} ds \; \frac{1}{\lambda^2} G\left(\frac{t-s}{\lambda^2}\right) e^{i (\lambda^{-2} H_S^{(0)} +  H_S^{(2)}) (t-s)}  V_{\lambda}(\lambda^{-2} s) 
		\end{align*}
		Taking into account \eqref{eq:defOfW} we obtain \eqref{eq:Weq}.
	\end{proof}
	
	Let us define the Laplace transforms of  $ W_{\lambda}(t) $ and $ G(t) $ as
	\begin{equation*}
		\tilde{W}_{\lambda}(p) = \int_0^{+\infty} dt e^{- p t} W_{\lambda}(t), \qquad \tilde{G}(p) = \int_0^{+\infty} dt e^{- p t} G(t).
	\end{equation*}
	
	\begin{lemma}
		The Laplace transform of the function  $ W_{\lambda}(t) $ has the form
		\begin{equation}\label{eq:WLaplaceTransform}
			\tilde{W}_{\lambda}(p) = \frac{1}{ p + \tilde{G}(- i H_S^{(0)} + \lambda^2( p - i  H_S^{(2)}))}
		\end{equation}
		which we understand as the function of self-adjoint matrix $ H_S \equiv H_S^{(0)} + \lambda^2 H_S^{(2)} $, which is defined in the usual way. Namely, let us diagonalize $ H_S = U \mathrm{diag}\;  \{E_{\alpha}\}  U^+ $ by unitary matrix $ U $, then
		\begin{equation*}
			\frac{1}{ p + \tilde{G}(- i H_S + \lambda^2 p)} = U \mathrm{diag}\;  \left\{ \frac{1}{ p + \tilde{G}(- i E_{\alpha} + \lambda^2 p)} \right\}  U^+ .
		\end{equation*}
	\end{lemma}
	
	\begin{proof}
		First of all let us calculate
		\begin{equation*}
			\int_0^{\infty} e^{- p t}G(t) e^{i H_S t} dt = U\operatorname{diag}  \left\{\int_0^{\infty} e^{- p t}G(t) e^{i E_{\alpha} t} dt \right\}  U^{\dagger} = U\operatorname{diag} \left\{  \tilde{G}(p - i E_{\alpha}) \right\}  U^{\dagger} = \tilde{G}(p - i H_S),
		\end{equation*}
		then
		\begin{equation*}
			\int_0^{\infty} e^{- p t} \lambda^{-2}G(\lambda^{-2} t) e^{i \lambda^{-2} H_S t} dt =  \tilde{G}(\lambda^2 p - i H_S).
		\end{equation*}
		Now, let us apply the Laplace transform to both sides of \eqref{eq:Weq}
		\begin{equation*}
			p  \tilde{W}_{\lambda}(p) -  W_{\lambda}(0) = - \tilde{G}(\lambda^2 p - i H_S^{(0)} - i \lambda^2  H_S^{(2)}) \tilde{W}_{\lambda}(p) .
		\end{equation*}
		Taking into account the initial condition $  W_{\lambda}(0) = I $ we obtain \eqref{eq:WLaplaceTransform}.
	\end{proof}
	
	For the next theorem we need the definition of a difference derivative of the function $ f(x) $ (see \cite[Theorem I.3]{Nazaikinskii2011}):
	\begin{equation}\label{eq:diffQuotDef}
		\frac{\delta f}{\delta x}(x,y) \equiv
		\begin{cases}
			\frac{f(x) - f(y)}{x - y}, & x \neq y\\
			f'(x) & x = y.
		\end{cases}
	\end{equation}
	We also need the definition of function $ f(x_1, \ldots, x_n) $ for non-commutative Hermitian matrices $ A_1, \ldots, A_n$ ordered by Feynman indices \cite[p. 26]{Nazaikinskii2011}. Let $ A_k $ have spectral decompositions $ A_k = \sum_{a_k} a_k \Pi_{a_k} $, then
	\begin{equation*}
		f(\overset{i_1}{A_1}, \ldots, \overset{i_n}{A_n}) \equiv \sum_{a_1, \ldots, a_n} f(a_1, \ldots, a_n) \; \mathrm{ord} \;\overset{i_1}{\Pi_{a_1}} \ldots \overset{i_n}{\Pi_{a_n}},
	\end{equation*}
	where $ \mathrm{ord} \;\overset{i_1}{\Pi_{a_1}} \ldots \overset{i_n}{\Pi_{a_n}} $ is product ordered in a such way, that the projectors with smaller indices stand to the left of the ones with larger indices. For example,
	\begin{equation}\label{eq:threeOperOrd}
		f(\overset{2}{A_1}, \overset{1}{A_2}, \overset{3}{A_3}) = \sum_{a_1, a_2, a_3}  f(a_1, a_2, a_3) \Pi_{a_3} \Pi_{a_1}\Pi_{a_2}.
	\end{equation}
	
	We are interested only in the first correction to the standard weak coupling limit as it is of most applied interest \cite{Jang2002, Trushechkin19a}. So in this paper we do not go to the further corrections, but let us remark that it seems to be possible in the way discussed in \cite[Appendix B]{Teretenkov2020nonpert}. 
	
	\begin{theorem}\label{th:assimpt}
		Let $  \tilde{G}(p) $  be twice continuously differentiable with respect to $ p $. For fixed $ t>0 $ at $ \lambda \rightarrow 0 $ one has
		\begin{equation}\label{eq:Wassymp}
			W_{\lambda}(t) = e^{L t} r + O(\lambda^4),
		\end{equation}
		where $ r $ and $ L $ are $ N \times N $ matrices defined as 
		\begin{equation*}
			r = 1  -\lambda^2   \tilde{G}'(- i H_S^{(0)}) 
		\end{equation*}
		and
		\begin{equation*}
			L = - \tilde{G}(- i H_S^{(0)}) + \lambda^2 \left(  \tilde{G}'(- i H_S^{(0)}) \tilde{G}(- i H_S^{(0)}) + i   \overset{2}{H_S^{(2)}}\frac{\delta \tilde{G}}{\delta p} (- i \overset{1}{H_S^{(0)}}, - i \overset{3}{H_S^{(0)}} )\right) .
		\end{equation*}
	\end{theorem}
	
	\begin{proof}
		By Theorem I.8 from \cite{Nazaikinskii2011} it is possible to expand the denominator of \eqref{eq:WLaplaceTransform} as
		\begin{equation*}
			(\tilde{W}_{\lambda}(p))^{-1}  =  p + \tilde{G}(- i H_S^{(0)}) + \lambda^2 ( p - i   \overset{2}{H_S^{(2)}})\frac{\delta \tilde{G}}{\delta p} (- i \overset{1}{H_S^{(0)}}, - i \overset{3}{H_S^{(0)}} )  + O(\lambda^4)
		\end{equation*}
		As 
		\begin{align*}
			( p - i   \overset{2}{H_S^{(2)}})\frac{\delta \tilde{G}}{\delta p} (- i \overset{1}{H_S^{(0)}}, - i \overset{3}{H_S^{(0)}} )  &= p \frac{\delta \tilde{G}}{\delta p} (- i \overset{1}{H_S^{(0)}}, - i \overset{3}{H_S^{(0)}} ) - i   \overset{2}{H_S^{(2)}}\frac{\delta \tilde{G}}{\delta p} (- i \overset{1}{H_S^{(0)}}, - i \overset{3}{H_S^{(0)}} ) \\
			&=  p \tilde{G}'(- i H_S^{(0)})  - i   \overset{2}{H_S^{(2)}}\frac{\delta \tilde{G}}{\delta p} (- i \overset{1}{H_S^{(0)}}, - i \overset{3}{H_S^{(0)}} )
		\end{align*}
		we have
		\begin{equation*}
			(\tilde{W}_{\lambda}(p))^{-1} = p \left(1  + \lambda^2   \tilde{G}'(- i H_S^{(0)})  \right) + \tilde{G}(- i H_S^{(0)})  -  i \lambda^2    \overset{2}{H_S^{(2)}}\frac{\delta \tilde{G}}{\delta p} (- i \overset{1}{H_S^{(0)}}, - i \overset{3}{H_S^{(0)}} ) + O(\lambda^4).
		\end{equation*}
		Taking into account
		\begin{equation*}
			p  + \left(1  + \lambda^2   \tilde{G}'(- i H_S^{(0)})  \right)^{-1} \left(\tilde{G}(- i H_S^{(0)})  -  i \lambda^2    \overset{2}{H_S^{(2)}}\frac{\delta \tilde{G}}{\delta p}  (- i \overset{1}{H_S^{(0)}}, - i \overset{3}{H_S^{(0)}} )\right) 
			= p + L + O(\lambda^4)
		\end{equation*}
		and
		\begin{equation*}
			\left(1  + \lambda^2   \tilde{G}'(- i H_S^{(0)})  \right)  = r^{-1}  + O(\lambda^4)
		\end{equation*}
		we obtain
		\begin{equation*}
			(\tilde{W}_{\lambda}(p))^{-1} = r^{-1} (p - L) + O(\lambda^4).
		\end{equation*}
		Then
		\begin{equation*}
			\tilde{W}_{\lambda}(p) =  (p - L)^{-1} r + O(\lambda^4)
		\end{equation*}
		After the inverse Laplace transform we obtain \eqref{eq:Wassymp}. 
	\end{proof}
	
	Let us also write $ L $ in a bit more explicit form. Let us expand $  H_S^{(0)} $ in spectral decomposition as 
	\begin{equation*}
		H_S^{(0)} = \sum_E E \, \Pi_E,
	\end{equation*}
	where $ E $ are its eigenvalues and $  \Pi_E $ are its eigenprojectors. Then taking into account \eqref{eq:diffQuotDef} and \eqref{eq:threeOperOrd} we obtain
	
	\begin{align*}
		L = - \sum_{E}\tilde{G}(- i E) \Pi_{E} &+ \lambda^2  \sum_{E}\left( \tilde{G}(- i  E)  +  i \Pi_{E} H_S^{(2)} \Pi_{E}  \right) \tilde{G}'(- i  E)  \Pi_{E}\\
		& - \lambda^2  \sum_{E \neq E'} \frac{\tilde{G}(- i E) - \tilde{G}(- i E')}{E - E'}   \Pi_{E} H_S^{(2)} \Pi_{E'}.
	\end{align*}
	
	Let us denote the reduced density matrix in the interaction picture with rescaled time as $  \rho_{SI, \lambda}(t)  $. By Th.~\ref{th:assimpt} and  formula \eqref{eq:solByProp} for (fixed) $ t >0 $ we obtain
	\begin{equation*}
		\rho_{SI, \lambda}(t) = \rho_{SI, \lambda}(t)|_{\rm as}
		+ O(\lambda^4)
	\end{equation*}
	with
	\begin{equation}\label{eq:solRhoAss}
		\rho_{SI, \lambda}(t)|_{\rm as} =
		\begin{pmatrix}
			1 - \mathrm{Tr} \; e^{L t} r (\rho_{S}(0))_{ee} r^+ e^{L^+ t} &  (\rho_{S}(0))_{ge} r^+ e^{L^+ t}\\
			e^{L t} r (\rho_{S}(0))_{eg} & e^{L t} r (\rho_{S}(0))_{ee} r^+ e^{L^+ t}
		\end{pmatrix}.
	\end{equation}
	
	As in  \cite{Teretenkov2020nonpert} and \cite{Teretenkov2021St} this asymptotic expansion is not uniform in time. Namely, it is not valid at times $ t = O(\lambda^2) $, which leads to the initial layer phenomenon. In particular, if $ \tilde{G}'(- i H_S^{(0)}) \neq 0 $, then $ \rho_{SI, \lambda}(0) \neq \rho_{SI, \lambda}(0)|_{\rm as}  $,
	\begin{equation*}
		\rho_{SI, \lambda}(0)|_{\rm as} =
		\begin{pmatrix}
			1 - \mathrm{Tr} \; r (\rho_{S}(0))_{ee} r^+ &  (\rho_{S}(0))_{ge} r^+ \\
			r (\rho_{S}(0))_{eg} &  r (\rho_{S}(0))_{ee} r^+ 
		\end{pmatrix} = R(\rho_{S}(0)),
	\end{equation*}
	where $ R $ is a linear superoperator defined by
	\begin{equation}\label{eq:renormSuperOp}
		R(\rho) =
		\begin{pmatrix}
			\rho_{gg} + \mathrm{Tr} \; (\rho_{ee}- r \rho_{ee} r^{+}) &  \rho_{ge} r^{+} \\
			r\rho_{eg} & r \rho_{ee} r^{+}
		\end{pmatrix}.
	\end{equation}
	The fact that our expansion is not valid at times $ t = O(\lambda^2) $ is what we mean exactly when saying that it is valid after the bath correlation time.
	
	Let $ L $ be a dissipative matrix, i.e. $ \mathrm{Re} \; v^+ L v \leqslant 0 $, then  \eqref{eq:solRhoAss} could be represented (see \cite{Teretenkov19OnePart}) as
	\begin{equation*}
		\rho_{SI, \lambda}(t)|_{\rm as}= e^{\mathcal{L} t} \rho_{SI, \lambda}(0)|_{\rm as},
	\end{equation*}
	where $ \mathcal{L}  $ has the GKSL form. The explicit form of $ \mathcal{L}  $ could be obtained as follows. Let $ | l\rangle $ be eigenvectors of $ L $ with eigenvalues $ - i \varepsilon_{l} - \Gamma_l/2$, $ \Gamma_l \geqslant 0 $, then
	\begin{equation}\label{eq:GKSLform}
		\mathcal{L}(\rho) =  - i \left[ \sum_l\varepsilon_{l} | l \rangle \langle l |, \rho \right] + \sum_l \Gamma_l \left( |0 \rangle \langle l | \rho |l \rangle \langle 0 | - \frac12 \{|l \rangle \langle 0  |0 \rangle \langle l | , \rho \} \right).
	\end{equation}
	
	Let us show that the   $ \mathrm{Re} \; v^+ L v \leqslant 0 $ is generally held with minor additional conditions. As for $ \lambda \rightarrow + 0 $ we should recover the usual weak coupling GKSL equation, then $ \mathrm{Re} \; v^+ L v \leqslant = - \mathrm{Re} \; v^+ \tilde{G}(- i H_S^{(0)}) v \leqslant 0 $ and, hence, $ \mathrm{Re} \;\tilde{G}(- i E) \geqslant 0 $ for all eigenvalues $ E $ of $ H_S^{(0)} $. (It is possible to obtain it in a more direct way, assuming $ G(t) $ is a correlation function, but this ''corrections based'' way of thinking is more natural for our discussion.) The case $ \mathrm{Re} \;\tilde{G}(- i E) = 0 $ for some $ E $ we would consider as exceptional and assume $ \mathrm{Re} \;\tilde{G}(- i E)  \neq  0 $. Thus, generally $  - \mathrm{Re} \; v^+ \tilde{G}(- i H_S^{(0)}) v > 0 $ for $ v \neq 0$ and, hence, $ \mathrm{Re} \; v^+ L v \leqslant 0 $ is held for sufficiently small $ \lambda $ as well.
	
	So let us summarize the results of this section. The reduced density matrix in the interaction picture with Bogolubov-van Hove scaling could be represented for $ t>0 $
	\begin{equation*}
		\rho_{SI, \lambda}(t) =  e^{\mathcal{L} t} R(\rho_{S}(0))
		+ O(\lambda^4), \qquad \lambda \rightarrow 0
	\end{equation*}
	where $ R $ is a ''renormalization'' superoperator defined by \eqref{eq:renormSuperOp} and $ \mathcal{L} $ has GKSL form \eqref{eq:GKSLform} for sufficiently small $ \lambda $ and is fully defined by $ L $ given in Th. \ref{th:assimpt}. So one should renormalize initial condition and then the dynamics is described by Markovian master equations. Let us also note that it is possible to act in reverse order. As $ r^{-1} L r = L + O(\lambda^4) $, then $  e^{L t} r = r e^{r^{-1} L r t} = r  e^{L t}  + O(\lambda^4) $, which leads to
	\begin{equation*}
		\rho_{SI, \lambda}(t) =  R(e^{\mathcal{L} t} \rho_{S}(0))
		+ O(\lambda^4),
	\end{equation*}
	so it is possible to evolve the initial density matrix in a Markovian way and renormalize the result after it.

	\section{System correlation function}
	\label{sec:correlation}
	
	The GKSL form of generator is sometimes considered \cite{Chruscinski2019} as a quantum definition of quantum Markovianity. Nevertheless, even in the classical case \cite{Feller1959} the Markovian form of the master equation is not enough for Markovianity of a classical stochastic process. The discussion in \cite[Subsection 5.1.3]{Li2018} shows that the condition that the correlation functions satisfy the generalized regression formulae is one of most natural generalizations of conditions which define classical Markov processes. For simplicity following \cite{Teretenkov2020nonpert,Gullo14} we discuss mostly the two-time correlation functions.
	
	First of all let us define $ \Phi_{t_1}^{t_2} $ for $ t_2 \geqslant  t_1$  as a linear superoperator such that
	\begin{equation*}
		\rho_{SI}(t_2) = \Phi_{t_1}^{t_2}(\rho_{SI}(t_1)),
	\end{equation*}
	which leads to
	\begin{gather}
		\Phi_{t_1}^{t_2}(\rho) = \nonumber\\
		\label{eq:divEvolMap}
		\begin{pmatrix}
			\rho_{gg} + \mathrm{Tr} \; (\rho_{ee}-V(t_2) (V(t_1))^{-1} \rho_{ee} (V(t_2) (V(t_1))^{-1})^+ &  \rho_{ge} (V(t_2) (V(t_1))^{-1})^{+} \\
			V(t_2) (V(t_1))^{-1} \rho_{eg} & V(t_2) (V(t_1))^{-1} \rho_{ee} (V(t_2) (V(t_1))^{-1})^+
		\end{pmatrix}.
	\end{gather}
	Here we assume, that $ V(t) $ is invertible as  it is impossible to write even a time-dependent GKSL equation for $ \rho_{SI}(t)  $ otherwise and the dynamics is non-Markovian at the level of the master equation already.

	Similar to  \cite{Teretenkov2020nonpert,Gullo14}  we are interested in some special correlation functions of system dipole operators as they play the most important role in spectroscopy and, hence, in experimental characterization of non-Markovianity. So let us define a dipole operator of the form
	\begin{equation*}
		\sigma_h \equiv
		\begin{pmatrix}
			0 & 0\\
			h & 0
		\end{pmatrix},
	\end{equation*}
	where $ h $ is an $ N $-dimensional vector.
	
	Then let us introduce Markovian correlation functions by quantum regression formula \cite[Subsection 3.4.1]{Li2018}.
	\begin{equation*}
		\langle \sigma_{h_2}^{\dagger}(t_2)  \sigma_{h_1}(t_1) \rangle_M \equiv \mathrm{Tr} \; \sigma_{h_2}^{\dagger} \Phi_{t_1}^{t_2} (\sigma_{h_1} \Phi_{0}^{t_1}(|0 \rangle \langle 0 |)) 
	\end{equation*}
	
	This result is assumed by Markovian approximation. But we will show that such a correlation function could be calculated for this model exactly. Let us also define the exact correlation function by
	\begin{equation*}
		\langle \sigma_{h_2}^{\dagger}(t_2)  \sigma_{h_1}(t_1) \rangle \equiv \mathrm{Tr}\; (\sigma_{h_2}^{\dagger} \otimes I) \mathcal{U}_{t_1}^{t_2} ((\sigma_{h_1} \otimes I) \mathcal{U}_{0}^{t_1}(|0 \rangle \langle 0 | \otimes  | \Omega \rangle \langle \Omega | )),
	\end{equation*}
	where $ \mathcal{U}_{t_1}^{t_2}  $ is unitary dynamics of states in the interaction picture.  
	
	\begin{theorem}\label{th:correlFun}
		For arbitrary $ h_1,. h_2 \in \mathbb{C}^N $ we have
		\begin{align*}
			\langle \sigma_{h_2}^{\dagger}(t_2)  \sigma_{h_1}(t_1) \rangle_M &= h_2^{+} V(t_2) (V(t_1))^{-1} h_1, \\
			\langle \sigma_{h_2}^{\dagger}(t_2)  \sigma_{h_1}(t_1) \rangle &= h_2^{+}  V(t_2-t_1) h_1.
		\end{align*}
	\end{theorem}
	
	\begin{proof}
		1) By \eqref{eq:divEvolMap}  we have $  \Phi_{0}^{t_1}(|0 \rangle \langle 0 |) = |0 \rangle \langle 0 | $, then taking into account $   \sigma_{h_1}^{\dagger} |0 \rangle \langle 0 | = \sigma_{h_1}^{\dagger}  $ and applying \eqref{eq:divEvolMap} once again we obtain
		\begin{align*}
			\mathrm{Tr} \; \sigma_{h_2}^{\dagger} \Phi_{t_1}^{t_2} (\sigma_{h_1} \Phi_{0}^{t_1}(|0 \rangle \langle 0 |)) &= \mathrm{Tr} \; \sigma_{h_2}^{\dagger} \Phi_{t_1}^{t_2} (\sigma_{h_1}) = 
			\mathrm{Tr} \;
			\begin{pmatrix}
				0 & h_2^+\\
				0 & 0
			\end{pmatrix}
			\begin{pmatrix}
				0 & 0\\
				V(t_2) (V(t_1))^{-1} h_1 & 0
			\end{pmatrix} \\
			&= \mathrm{Tr} \;
			\begin{pmatrix}
				h_2^+ V(t_2) (V(t_1))^{-1} h_1 & 0\\
				0 & 0
			\end{pmatrix} = h_2^+ V(t_2) (V(t_1))^{-1} h_1 
		\end{align*}
		2) The unitary evolution of pure states $ \mathcal{U}_{t_1}^{t_2} $ in zero- and one-particle subspaces could be found in \cite{Teretenkov19}. Namely, we have
		\begin{equation*}
			\mathcal{U}_{0}^{t_1}(|0 \rangle \langle 0 | \otimes  | \Omega \rangle \langle \Omega | )) = |0 \rangle \langle 0 | \otimes  | \Omega \rangle \langle \Omega | .
		\end{equation*}
		Then by direct calculation we obtain
		\begin{equation*}
			\sigma_{h_1} \otimes I |0 \rangle \langle 0 | \otimes  | \Omega \rangle \langle \Omega | =  ( 0 \oplus h_1)  \otimes | \Omega \rangle \langle 0 | \otimes \langle \Omega |
		\end{equation*}
		and
		\begin{equation*}
			\mathcal{U}_{t_1}^{t_2}( ( 0 \oplus h_1)  \otimes | \Omega \rangle \langle 0 | \otimes \langle \Omega | ) = U_{t_1}^{t_2} ( 0 \oplus h_1)  \otimes | \Omega \rangle \langle 0 | \otimes \langle \Omega |  (U_{t_1}^{t_2})^{\dagger}.
		\end{equation*}
		From \cite{Teretenkov19} we have
		\begin{equation*}
			\langle 0 | \otimes \langle \Omega |  (U_{t_1}^{t_2})^{\dagger} =\langle 0 | \otimes \langle \Omega |
		\end{equation*}
		and
		\begin{equation*}
			U_{t_1}^{t_2} ( 0 \oplus h_1)  \otimes | \Omega \rangle = 0 \oplus V(t_2-t_1)h_1  \otimes | \Omega \rangle + |0 \rangle \otimes |\chi(t_2 - t_1)\rangle,
		\end{equation*}
		where $ |\chi(t_2 - t_1)\rangle $ is a one-particle state of the reservoir, explicit form of which is not important for us (we only need $ \langle \Omega |\chi(t_2 - t_1)\rangle  = 0$).
		Hence,
		\begin{align*}
			U_{t_1}^{t_2} ( 0 \oplus h_1)  \otimes | \Omega \rangle \langle 0 | \otimes \langle \Omega |  (U_{t_1}^{t_2})^{\dagger} 
			&=  0 \oplus V(t_2-t_1)h_1 \langle 0 | \otimes | \Omega \rangle  \langle \Omega |  + | 0 \rangle \langle 0 | \otimes | \chi(t_2 - t_1) \rangle  \langle \Omega | \\
			&=  \begin{pmatrix}
				0 & 0\\
				V(t_2-t_1)h_1 & 0
			\end{pmatrix} \otimes  | \Omega \rangle  \langle \Omega |  + | 0 \rangle \langle 0 | \otimes | \chi(t_2 - t_1) \rangle  \langle \Omega |.
		\end{align*}
		and finally we obtain
		\begin{equation*}
			\mathrm{Tr} \;
			\begin{pmatrix}
				0 & h_2^+\\
				0 & 0
			\end{pmatrix} \otimes I \left(
			\begin{pmatrix}
				0 & 0\\
				V(t_2-t_1)h_1 & 0
			\end{pmatrix} \otimes  | \Omega \rangle  \langle \Omega |  + | 0 \rangle \langle 0 | \otimes | \chi(t_2 - t_1) \rangle  \langle \Omega |
			\right) 
			= h_2^{+}  V(t_2-t_1) h_1	.
		\end{equation*}
	\end{proof}
	
	So as in \cite{Gullo14} one regards the condition $ 
	\langle \sigma_{h_2}^{\dagger}(t_2)  \sigma_{h_1}(t_1) \rangle = 	\langle \sigma_{h_2}^{\dagger}(t_2)  \sigma_{h_1}(t_1) \rangle_M $ as a part of the definition of quantum Markovianity, then it is equivalent to $  V(t_2) = V(t_2-t_1) V(t_1)$ for all $ t_2 \geqslant t_1 \geqslant $, i.e. to the semigroup property. Hence, strictly speaking, it is helds only in the zero order of perturbation theory if $ r \neq I $, i.e. if $ \tilde{G}'(- i H_S^{(0)})  \neq 0 $. Explicitly for $ t_2 > t_1 >0  $  we have
	\begin{align*}
		\langle \sigma_{h_2}^{\dagger}(t_2)  \sigma_{h_1}(t_1) \rangle_M &= h_2^{+}  e^{L (t_2 - t_1)} h_1 + O(\lambda^4), \\
		\langle \sigma_{h_2}^{\dagger}(t_2)  \sigma_{h_1}(t_1) \rangle &= h_2^{+}  e^{L (t_2 - t_1)} r h_1 + O(\lambda^4).
	\end{align*}
	So all the non-Markovinity occurs from the initial renormalization and could be absorbed in renormalization of the correlation functions
	\begin{equation*}
		\langle \sigma_{h_2}^{\dagger}(t_2)  \sigma_{h_1}(t_1) \rangle_{r} \equiv 	\langle \sigma_{h_2}^{\dagger}(t_2)  \sigma_{r^{-1}h_1}(t_1) \rangle, \qquad t_2 > t_1 >0,
	\end{equation*}
	which leads to $ \langle \sigma_{h_2}^{\dagger}(t_2)  \sigma_{h_1}(t_1) \rangle_{r} = 	\langle \sigma_{h_2}^{\dagger}(t_2)  \sigma_{h_1}(t_1) \rangle_M  + O(\lambda^4) $. On the one hand, the renormalization $ \langle \sigma_{h_2}^{\dagger}(t_2)  \sigma_{h_1}(t_1) \rangle_{r} $ generalizes the results of \cite[Section 3]{Teretenkov2020nonpert}, on the other hand, it shows that in general it is impossible just to rescale correlation functions by a constant and one should consider the linear combinations of the non-renormailized correlation functions instead as $ r $ is a matrix  rather than constant now.
	
	Let us also consider three-time correlation functions which are widely used in the 2-dimensional echo spectroscopy to measure the population dynamics \cite{Cho2009}:
	\begin{align*}
		& \langle \sigma_{h_2}^{\dagger}(\tau)  \sigma_{h_4}(T+\tau + t)  \sigma_{h_3}^{\dagger}(T+\tau)  \sigma_{h_1} \rangle_M \equiv  \mathrm{Tr} \; \sigma_{h_4} \Phi_{\tau + T}^{\tau + T + t}(\sigma_{h_3}^{\dagger} \Phi_{\tau}^{\tau + T} ( \Phi_{0}^{\tau} ( \sigma_{h_1} |0 \rangle \langle 0 |) \sigma_{h_2}^{\dagger}  ) ), \\
		&\langle \sigma_{h_2}^{\dagger}(\tau)  \sigma_{h_4}(T+\tau + t)  \sigma_{h_3}^{\dagger}(T+\tau)  \sigma_{h_1} \rangle \\  
		& \equiv \mathrm{Tr} \; (\sigma_{h_4} \otimes I) \mathcal{U}_{\tau + T}^{\tau + T + t}((\sigma_{h_3}^{\dagger} \otimes I) \mathcal{U}_{\tau}^{\tau + T} ( \mathcal{U}_{0}^{\tau} ( (\sigma_{h_1} \otimes I) |0 \rangle \langle 0 | \otimes  | \Omega \rangle \langle \Omega |) (\sigma_{h_2}^{\dagger} \otimes I)  ) ),
	\end{align*}
	for $ \tau \geqslant 0 $, $ T \geqslant 0 $ and $ t \geqslant 0 $.
	
	Similar to Th.~\ref{th:correlFun} it is possible show that 
	\begin{align*}
		\langle \sigma_{h_2}^{\dagger}(\tau)  \sigma_{h_4}(T+\tau + t)  \sigma_{h_3}^{\dagger}(T+\tau)  \sigma_{h_1} \rangle_M &=  h_3^+ V(\tau + T) h_1 h_2^+ (V^+(\tau))^{-1} V^+(t + T + \tau) h_4,\\
		\langle \sigma_{h_2}^{\dagger}(\tau)  \sigma_{h_4}(T+\tau + t)  \sigma_{h_3}^{\dagger}(T+\tau)  \sigma_{h_1} \rangle &=  h_3^+ V(\tau + T) h_1 h_2^+ V^+(t + T) h_4.
	\end{align*} 
	
	So they also coincide only in the case, when $ V(t) $ is a semigroup, but after the bath correlation time it could be compensated by renormalization of this correlation function similar to that for two-time correlation functions.
	
	\section{Conclusions}
	
	For our model we have obtained the corrections to usual weak coupling limit reduced dynamics. We have shown that after the bath correlation time the dynamics could be described by the Markovian master equation, but  either the initial condition or the final result should be renormalized. The correlation functions do not satisfy the Markovian formulae, but do satisfy them after the renormalization. We call such a behavior of the reduced dynamics long-time Markovian.   
	
	We think that our results are important for two main reasons. The first one is for general development of corrections to Markovian master equations and regression formulae. The results of \cite{Trushechkin2021Steklov} suggest that time-indpendent master equations could be derived in a much more general situation than this specific model. The second one is that they show the dynamical viewpoint on Markovianity should be developed. We think that it is also supported by recent results \cite{Burgarth2021}. The authors of  \cite{Burgarth2021} interpret them as impossibility to detect non-Markovianity by some initial region of evolution due to the fact that their example exhibits Markovian dynamics up to some fixed time. But we think it is more natural just to say  that the dynamics is Markovian up to this fixed time and becomes non-Markovian after that.

	\section{Acknowledgments}
	The author thanks  A.\,S.~Trushechkin for the fruitful discussion which led to the ideas of some problems considered in the work.

\end{document}